\documentclass[journal,onecolumn]{IEEEtranTIE}

\usepackage{empheq}

\usepackage{amsfonts}
\usepackage{amsmath}
\usepackage{amssymb}
\usepackage{graphicx}
\usepackage{float}
\usepackage[mathscr]{eucal}
\usepackage{amsfonts}
\usepackage{float}
\usepackage[colorlinks,linkcolor=blue]{hyperref}
\usepackage[latin1]{inputenc}
\usepackage{amsmath}
\usepackage{graphicx}
\usepackage{amssymb}
\usepackage{amsthm}
\usepackage{verbatim}
\usepackage{epsfig}
\usepackage[labelfont=bf]{caption}
\usepackage{enumerate}
\usepackage{subfigure}
\usepackage{epstopdf}

\newtheorem{lemma}{Lemma}
\newtheorem{proposition}{Proposition}
\newtheorem{corollary}{Corollary}
\newtheorem{fact}{Fact}
\newtheorem{definition}{Definition}
\newtheorem{remark}{Remark}

\newtheorem{assumption}{Assumption}

\def\begcen{\begin{center}}
\def\endcen{\end{center}}

\def\begdef{\begin{definition}}
\def\enddef{\end{definition}}

\newcommand{\col}{ \mbox{col} }

\def\caly{{\cal Y}}
\def\calg{{\cal G}}

\def\calw{{\cal W}}

\def\calf{{\cal F}}
\def\cald{{\cal D}}

\def\liminf{\lim_{t \to \infty}}

\def\L2{{\cal L}_2}
\def\L2e{{\cal L}_{2e}}

\def\rea{\mathbb{R}}

\def\adj{\mbox{adj}}

\def\begmat#1{\begin{bmatrix}#1\end{bmatrix}}
\def\begali#1{\begin{align}{#1}\end{align}}
\def\begalis#1{\begin{align*}{#1}\end{align*}}

\def\begequarr{\begin{eqnarray}}
\def\endequarr{\end{eqnarray}}
\def\begequarrs{\begin{eqnarray*}}
\def\endequarrs{\end{eqnarray*}}
\def\begarr{\begin{array}}
\def\endarr{\end{array}}
\def\begequ{\begin{equation}}
\def\endequ{\end{equation}}
\def\lab{\label}
\def\begdes{\begin{description}}
\def\enddes{\end{description}}
\def\begenu{\begin{enumerate}}
\def\begite{\begin{itemize}}
\def\endite{\end{itemize}}
\def\endenu{\end{enumerate}}

\def\lef[{\left[\begin{array}}
\def\rig]{\end{array}\right]}

\def\begcen{\begin{center}}
\def\endcen{\end{center}}
\def\begrem{\begin{remark}\rm}
\def\endrem{\end{remark}}
\def\begassum{\begin{assumption}}
\def\endassum{\end{assumption}}
\def\begassums{\begin{assumption*}}
\def\endassums{\end{assumption*}}
\def\begassu{\begin{ass}}
\def\endassu{\end{ass}}
\def\beglem{\begin{lemma}}
\def\endlem{\end{lemma}}
\def\begcor{\begin{corollary}}
\def\endcor{\end{corollary}}
\def\begfac{\begin{fact}}
\def\endfac{\end{fact}}
\def\TAC{{\it IEEE Trans. Automat. Contr.}}
\def\AUT{{\it Automatica}}

\def\SCL{{\it Systems and Control Letters}}

\def\liminf{\lim_{t \to \infty}}

\def\L2e{{\cal L}_{2e}}

\def\rea{\mathbb{R}}

\def\adj{\mbox{adj}}
\def\col{\mbox{col}}

\def\et{\varepsilon_t}

\def\TAC{{\it IEEE Trans. Automatic Control}}
\def\TIE{{\it IEEE Trans. Industrial Electronics}}

\def\SCL{{\it Systems \& Control Letters}}
\def\AUT{{\it Automatica}}

\def\TPE{{\it IEEE Trans. on Power Electronics}}

\newcommand{\tarc}{\mbox{$\frown$}}
\newcommand{\arc}[1]{\stackrel{\tarc}{#1}}

\def\begsubequ{\begin{subequations}}
\def\endsubequ{\end{subequations}}

\usepackage{xcolor}

\graphicspath{{figs/}{img/}}

\begin{document}
	\title{	On-line Parameter Estimation of the Polarization Curve of a Fuel Cell with Guaranteed Convergence Properties: Theoretical and Experimental Results}
	
	\author{
		\vskip 1em
		
		Carlo Beltr\'an, 
		Alexey Bobtsov,
		Romeo Ortega,
		Diego Langarica-C\'ordoba,\\
		Rafael Cisneros, 
		Luis H. D\'iaz-Saldierna
		
		\thanks{
			
			
			C. Beltr\'an and D. Langarica are with the Faculty of Sciences, Autonomous
			University of San Luis Potosí (UASLP), San Luis Potos\'i, 78295, Mexico (e-mails: carlo.beltran@ieee.org, diego.langarica@uaslp.mx). 
			
			A. Bobstov is with the Department of Control Systems and Robotics, ITMO University, 197101, Saint Petersburg, Russia (e-mail: bobtsov@mail.ru). 
			
			R. Ortega and R. Cisneros are with the Electrical and Electronic Engineering Department, Instituto Tecnol\'ogico Aut\'onomo de M\'exico (ITAM), Mexico City, 01080, Mexico (e-mails:romeo.ortega@itam.mx, rcisneros@itam.mx).
			
			L. D\'iaz is with the Instituto Potosino de Investigaci\'on Cient\'ifica y Tecnol\'ogica,  San Luis Potos\'i, 78216, Mexico (e-mail: ldiaz@ipicyt.edu.mx).
		}
	}
	\maketitle
\begin{abstract}
	In this paper we address the problem of online parameter estimation of  a Proton Exchange Membrane Fuel Cell (PEMFC) polarization curve, that is, the static relation between the voltage and the current of the PEMFC. The task of designing this estimator---even off-line---is complicated by the fact that the uncertain parameters enter the curve in a highly nonlinear fashion, namely in the form of {\em nonseparable} nonlinearities. We consider several scenarios for the model of the polarization curve, starting from the standard {\em full model} and including several popular {\em simplifications} to this complicated mathematical function. In all cases, we derive {\em separable} regression equations---either linearly or nonlinearly parameterized---which are instrumental for the implementation of the parameter estimators. We concentrate our attention in  {\em on-line} estimation schemes for which, under suitable excitation conditions, global {\em parameter convergence} is ensured.  Due to this global convergence properties the estimators are robust to unavoidable additive noise and structural uncertainty. Moreover, their on-line nature endow the schemes with the ability to track (slow) parameter variations, that occur during the operation of the PEMFC. These two features---unavailable in time-consuming off-line data-fitting procedures---make the proposed estimators  helpful for an on-line time-saving characterization of a given PEMFC, and the implementation of fault-detection procedures and model-based adaptive control strategies.   Simulation and {\em experimental} results that validate the theoretical claims are presented.  
\end{abstract}
\begin{IEEEkeywords}
Proton Exchange Membrane Fuel Cell, Online estimation, Parameter identification
\end{IEEEkeywords}
\markboth{IEEE TRANSACTIONS ON INDUSTRIAL ELECTRONICS}%
{}

\definecolor{limegreen}{rgb}{0.2, 0.8, 0.2}
\definecolor{forestgreen}{rgb}{0.13, 0.55, 0.13}
\definecolor{greenhtml}{rgb}{0.0, 0.5, 0.0}

%
\section{Introduction}
\lab{sec1}
%
Proton Exchange Membrane Fuel Cells (PEMFCs) have emerged as a promising and scalable energy conversion and sustainable power generation technology. As a subset of fuel cell technology, PEMFCs offer a variety of unique features that make them suitable for diverse applications across various industries such as electrical transportation and power generation. Among others, energy efficiency, rapid start-up, compact size and weight, low noise, low emission and modularity characteristics have widespread this type of fuel cells \cite{alyakhni2021comprehensive}. However, the determination of the state of health of a functioning Proton Exchange Membrane Fuel Cell (PEMFC), the optimization of the operating conditions and its energy management are topics of practical importance \cite{kim2016onlineapplicable}.  Usually such investigations require complicated experiments and instrumentation that cannot be applied in field operations; therefore, time-consuming off-line data-fitting procedures are employed. In recent years mathematical models that describe the operation of PEMFC  under different operating conditions have been developed. Due to the complexity of the behavior of the PEMFC these models involve highly nonlinear relations that depend, again in a nonlinear way, on uncertain parameters.  The values of these parameters, which reflect the actual PEMFC performance, inevitably affect effectiveness of the developed models in simulation, design, optimal operation and control. Therefore, it is indispensable to investigate the problem of {\em parameter estimation} of these models. Several  unwanted phenomena, such as catalyst poisoning, flooding or drying of the membrane etc. can be monitored, and the working parameters of an PEMFC can be quickly regulated with the help of real time parameter estimation \cite{debeer2015online}---hence our interest in {\em on-line} parameter estimators. Also, to ensure good performance in the face of uncertainty and noise, we concentrate our attention on schemes for which we can prove {\em global convergence} properties under some suitable excitation conditions.

A particularly relevant practical model is the one describing the quasi-static behavior of the PEMFC, which is captured by the so-called {\em polarization curve} \cite{SHAetal}. This curve describes the behavior of the cell voltage, which is given by the equilibrium potential diminished by irreversible losses, often called polarizations, due
to three sources: (i) activation polarization, (ii) ohmic polarization, and (iii) concentration polarization. An universally accepted model of the polarization curve is given as \cite[Chapter 3]{LARDICbook}:
\begequ
\lab{polcur}
v_{fc}(i_{fc})=v_{rev}-v_{act}(i_{fc})-v_{ohm}(i_{fc})-v_{con}(i_{fc}),
\endequ
where $v_{fc}(t) \in \rea$ and $i_{fc}(t) \in \rea$ are the PEMFC voltage and current, respectively. The term $v_{rev}$ represents the reversible potential, which is not a function of the current and is given by the classical Nernst equation. Since $v_{rev}$ depends on the PEMFC temperature and partial pressure, its time scale is significantly slower than the other losses of electrical nature, and can be assumed to be constant. The term $v_{act}(i_{fc})$ is the activation polarization described using an empirical approximation of the Tafel equation. The term $v_{ohm}(i_{fc})$ captures the ohmic losses, and $v_{con}(i_{fc})$ is the concentration polarization. Classical models for the first two components are given as \cite[Eq. 3.14]{LARDICbook}, \cite[Subsection 6.2.1]{GUZSCIbook}:
\begsubequ
\lab{vacvohm}
\begali{
	\lab{vact}
	v_{act}(i_{fc}) & =\theta_4+\theta_1 \ln(i_{fc})\\
	\lab{vohm}
	v_{ohm}(i_{fc})& =\theta_2 i_{fc},
}
\endsubequ
while for the concentration term, there are two popular versions reported in the literature \cite[Equation (1)]{SHAetal}   and \cite[Equation (4)]{BELetal}: 
\begin{subequations}\label{vcon}
	\begin{empheq}[left={v_{con}(i_{fc})=\empheqlbrace\,}]{align}
		& \theta_5 e^{\theta_3 i_{fc}} \label{vcon1}\\
		& \theta_5(i_{fc})^{\theta_3}. \label{vcon2}
	\end{empheq}
\end{subequations}

As indicated in \cite{GUZ}, the model \eqref{vact} may be not valid for very small current densities, {\em i.e.}, where the influence of the activation polarization is dominant, and an empirical approximation of the Tafel equation may be preferred, for instance 
\begequ
\lab{vact1}
v_{act}(i_{fc})=\theta_4\left(1 - e^{\theta_1i_{fc}}\right).
\endequ 

All the parameters given in these functions---namely $\theta:=\col(\theta_1,\theta_2,\theta_3,\theta_4,\theta_5)$---are positive and {\em uncertain} and, in particular, in the concentration term \eqref{vcon} the unknown parameter enters in an exponential form, leading to a {\em nonseparable nonlinear parameterization} (NSNLP)---that is, a nonlinear parameterization that cannot be written as $\phi^\top(\cdot)\calw(\cdot)$, where the function $\phi$ depends on measurable signals and is independent of the {\em unknown parameter}, which appear only in the function $\calw$.\footnote{Notice that if we adopt for the activation term the expression \eqref{vact1} instead of \eqref{vact}, a second nonseparable parameterization is added. Since dealing with two nonseparable nonlinearities is still an open problem, this model is not considered in the paper.}  The task of estimating on-line the parameters of systems with NSNLP is essentially open, making the solution of the present problem unattainable with existing parameter estimation techniques. One of the contributions of this paper is the development of a procedure to transform the NSNLP into a {\em separable} one, to which we can apply some new parameter estimation techniques. In particular, the least-squares plus dynamic regression equation and mixing (LSD) estimator recently reported in \cite{ORTROMARA}, which can deal with {\em  separable nonlinear} parameterizations (SNLP). In the case, when the parameterization is linear we can apply simpler least-squares of gradient algorithms \cite{SASBODbook,TAObook}\\

In the interest of investigating the accuracy of the approximate models we consider in the paper four descriptions of the polarization curve:
\begenu[{\bf M1}]
\item The full model \eqref{polcur}, \eqref{vacvohm} with \eqref{vcon1} for the concentration term. 
\item Invoking some practical considerations discussed for instance in  \cite{BELetal,LEELALAPP}, we concentrate on $v_{con}(i_{fc})$ and assume that the other terms can be approximated by a {\em known} constant $E_{oc}$, which is the voltage at which the linearized curve crosses the $y$ axis at the no-current state, leading to  the simplified  model  
\begequ
\lab{m2}
v_{fc}(i_{fc})=E_{oc}-a e^{b i_{fc}},
\endequ
where $a>0$ and $b>0$ are some {\em unknown} parameters.
\item Adopting the same arguments invoked above, but taking \eqref{vcon2} for the concentration term, leads to
\begequ
\lab{m3}
v_{fc}(i_{fc})=E_{oc}-a \big(i_{fc}\big)^{b}.
\endequ
\item Considering that $\theta_3$ in \eqref{vcon} is a small number---compared with the other quantities---we assume it is {\em equal to zero} and define the function
\begequ
\label{m4}
v_{fc}(i_{fc})=\theta_6+\theta_1 \ln(i_{fc})+\theta_2 i_{fc},
\endequ 
where we defined $\theta_6:=\theta_1+\theta_5$.
\endenu

In models {\bf M2} and {\bf M3} we take into account  the well-known fact that  

\begin{subequations}\label{eoc}
	\begin{empheq}[left={E_{oc}>\empheqlbrace\,}]{align}
		& a e^{b i_{fc}} \label{eoc1}\\
		& a(i_{fc})^{b}. \label{eoc2}
	\end{empheq}
\end{subequations}

The design procedure that we follow to solve the parameter estimation problem consists of four steps.
\begenu[{\bf S1}]
\item Proceeding from one of the NSNLP mapping given in {\bf M1}-{\bf M3}, construct a new SNLP of the form 
\begequ
\lab{nlpre}
Y(t)=\phi^\top(t)\calw(\cdot),
\endequ
where the functions $Y(t)\in \rea$ and $\phi(t) \in \rea^p$ are measurable and $\calw:\rea^q \to \rea^p$  is a known {\em nonlinear} mapping whose argument is a $q$-dimensional vector containing the unknown parameters. Notice that for the case of {\bf M4} this step can be obviated, since the parameterization is already separable and, moreover, {\em linear}. 
\item Propose a parameter estimator for the SNLP \eqref{nlpre}. Due to the more relaxed signal excitation requirements, we implemented the LSD estimator of \cite{ORTROMARA} in both simulation and experimental results. However, for a linear regression, the well-known gradient descent algorithm  \cite{SASBODbook} can be alternatively implemented.
\item In the LSD estimator, a condition is imposed on the nonlinear mapping $\calw$ of the SNLP, namely that it is a {\em monotonizable} nonlinearity.  As explained in \cite{PAVetal}, a necessary and sufficient condition to ensure this property is the feasibility of an easily verifiable {\em linear matrix inequality} (LMI). 
\item Graphically compare the polarization curve resulting from the estimator using models {\bf M2}-{\bf M4} with the ``true'' curve. In the simulation  results  this curve is that of \textbf{M1} whereas in the experimental results it corresponds to the curve, experimentally obtained from the PEMFC system in our test bench.
\endenu

A key feature of the LSD estimator is that it ensures GEC imposing an {\em interval excitation} (IE) assumption \cite{TAObook} on the regressor  $\phi$ of \eqref{nlpre}, which is a considerable milder assumption compared with the classical {\em persistency of excitation} condition \cite{SASBODbook,TAObook} imposed in standard gradient or least-squares estimators. In spite of this fact, we have to point out that the estimation of the parameters of the full model {\bf M1} was {\em unsuccessful}. Indeed, in spite of the selection of ``highly exciting" input signals $i_{fc}$ it was not possible to meet the IE requirement of the LSD estimator for parameter convergence. As discussed in the paper, this major drawback is due to the fact that some of the elements of the regressor $\phi$ of the SNLP are {\em linearly dependent}. Interestingly, this phenomenon is absent in the approximate models  {\bf M2}-{\bf M4}, for which {\em global exponential convergence} (GEC), with all signals remaining bounded, is indeed ensured.  

In spite of our inability to consistently estimate the parameters of the full model {\bf M1}, we believe the derivation of an SNLP for this model is an important contribution to the field. Indeed, to the best of the authors' knowledge, this is the first time a GEC on-line estimator of the parameters of realistic models of the polarization curve is reported. The field being dominated by  standard least squares and extended Kalman filters \cite{ghaderi2019simultaneously,barragan2020iterative}, based on linear approximations of the nonlinear model, with no guarantee of parameter convergence. See \cite{KANetal} for a recent survey on the topic. \\

The remainder of the paper is organized as follows. In Section \ref{sec2} we derive the SNLP for the full model {\bf M1} and discuss the issue of lack of excitation of the associated regressor. Unfortunately, this derivation requires the measurement of the {\em time derivative} of $i_{fc}$. In Section \ref{sec3} we consider the simplified models \eqref{m2}-\eqref{m4}. The details of the proposed LSD and gradient estimators are given in Section \ref{sec4}. Simulation and experimental results that illustrate the theoretical developments are given in  Section \ref{sec5} and \ref{sec6}, respectively. We wrap-up the paper with concluding remarks and future research in Section \ref{sec7}. In Appendix \ref{appa} we derive a SNLP for the full model {\bf M1} in the particular case when the current is a {\em sinusoidal} signal, which is of interest in an off-line scenario. Interestingly, in these derivations we {\em obviate} the need for the measurement of the derivative of $i_{fc}$. To simplify the reading, in Appendix \ref{appb} we give a list of acronyms. \\

\noindent {\bf Notation.}  $I_n$ is the $n \times n$ identity matrix  For a column vector $a=\col(a_1,a_2,\dots,a_n) \in \rea^n$, we denote $|a|^2:=a^\top a$ and define $a_{i,j}:=\col(a_i,a_{i+1},\dots,a_j)$, for $i,j \in \{1,\dots,n\}$, with $n\geq j>i$. The action of a linear time-invariant (LTI) filter $\calf(p) \in \rea(p)$, with $p$ the derivative operator  $p^n(w)=:{d^n w(t)\over dt^n}$, on a signal $w(t)$ is denoted as $\calf(p)(w)$.  To simplify the notation, the arguments of all functions and mappings are written only when they are first defined and are omitted in the sequel.
%
\section{Derivation of a SNLP for the Full Model M1}
\lab{sec2}
%
In this section we derive a SNLP for the full the polarization curve \eqref{polcur}, \eqref{vacvohm} and \eqref{vcon1}---called {\bf M1} in Section \ref{sec1}. Unfortunately, this parameterization {\em does not} contain all of the paramters $\theta$, but only $\theta_{1,4}$. To estimate the remaining parameter $\theta_5$ it is necessary to invoke a certainty-equivalence argument to compute, using the estimate of the other parameters, an estimate of it via an elementary algebraic operation, whose description is also included here. As indicated in Section \ref{sec1}, unfortunately, the regressor of the SNLP derived here contains elements which are {\em linearly} dependent. Consequently, it is not possible to achieve the desired parameter convergence. This delicate issue is also discussed here. It may be argued that including a parameterization for which convergence is unattainable is a futile exercise, however, we believe it is of academic interest to illustrate with an example the potential appearance of the undesired linear dependence phenomenon, an issue that is often overlooked in the literature.    
\subsection{Derivation of a SNLP for $\theta_{1,4}$}
\lab{subsec21}
%
To simplify the notation we introduce the definitions
\begequ
\lab{yu}
y:=v_{fc},\;u:=i_{fc},
\endequ
and rewrite the polarization curve \eqref{polcur}, \eqref{vacvohm} and \eqref{vcon1} as
\begequ
\lab{yu1}
y=\theta_4+\theta_1 \ln(u)+\theta_2 u + \theta_5 e^{\theta_3 u}.
\endequ
To derive the SNLP we make the following assumption.\\

\noindent {\bf A1} Besides the measurable signals $u$ and $y$ it is possible to reconstruct the signal $\dot u$.\\

\begin{lemma}\em
\lab{lem1}
Consider the model of the fuel cell voltage-current static characteristic given in \eqref{yu}, \eqref{yu1} satisfying Assumption {\bf A1}. There exists measurable signals $Y_{M1}(t)\in \rea,\;\phi_{M1}(t) \in \rea^5$ such that the following SNLP holds:
\begin{equation}
\label{nlpre1}
Y_{M1}(t)= \phi_{M1}^\top(t) \calw(\theta_{1,4})+\et,\quad \forall t \geq 0,
\end{equation}
where we defined the nonlinear mapping $\calw:\rea^4 \to \rea^5$
\begequ
\lab{calw}
\calw(\theta_{1,4}):=\col(\theta_1,\theta_3,\theta_1\theta_3,\theta_2\theta_3,\theta_2-\theta_3\theta_4),
\endequ
and $\et$ is a signal exponentially decaying to zero due to the action of the various filters.\footnote{It has been shown in \cite[Lemma 1]{ARAetalmicnon} that the presence of the exponentially decaying term $\et$ does not affect the validity of our asymptotic behavior claims, hence it is discarded in the future.}
\end{lemma}
\begin{proof}
Differentiating \eqref{yu1} we get
\begalis{
\dot y & = \theta_1 p(\ln(u))+\theta_2 p(u) + \theta_3 \dot u \left(\theta_5 e^{\theta_3 u}\right)\\
& = \theta_1 p(\ln(u))+\theta_2 p(u) + \theta_3 \dot u \left[y -\theta_4- \theta_1 \ln(u)- \theta_2 u \right]\\
& = \theta_1 p(\ln(u))+(\theta_2- \theta_3 \theta_4)p(u) + \theta_3 \dot u \left[y - \theta_1 \ln(u)- \theta_2 u \right],
}
where we used \eqref{yu1} again to get the second identity.  Let us apply now to this equation the LTI filter
\begequ
\lab{calf}
\calf(p)={\lambda \over p+\lambda},
\endequ
where $\lambda>0$ is a tuning parameter, yielding
\begalis{
{\lambda p \over p+\lambda} (y) & = \theta_1 {\lambda p \over p+\lambda}(\ln(u))+( \theta_2 - \theta_3\theta_4){\lambda p \over p+\lambda}(u) +\theta_3 {\lambda \over p+\lambda}(\dot u y)-{\theta_2\theta_3 \over 2} {\lambda p\over p+\lambda}(u^2)- \theta_1\theta_3{\lambda \over p+\lambda}(\dot u \ln(u)) ,
}
where we have used the fact that $p(u^2)=2 u \dot u$ to obtain the fourth right hand term in the first identity. The proof is completed defining the regressor vector.
\begequ
\lab{phim1}
\phi_{M1}:=\col\left({\lambda p \over p+\lambda}(\ln(u)), {\lambda \over p+\lambda}(\dot u y), -{\lambda \over p+\lambda}(\dot u \ln(u)), -{1 \over 2} {\lambda p\over p+\lambda}(u^2), {\lambda p \over p+\lambda}(u)\right),
\endequ
and denoting 
\begali{
\lab{ym1}
Y_{M1} & :={\lambda p\over p+\lambda}(y).
}
\end{proof}
\subsection{Discussion}
\lab{subsec22}
%
\noindent {\bf D1} We make the important observation that the regression equation \eqref{nlpre} is nonlinear with respect to the physical parameters $\theta_{1,4}$. Although it is possible to {\em overparameterize} the SNLP to obtain a linear one, it is well-known that this procedure has serious deleterious effects---see \cite{ORTetalaut21} for a detailed discussion on this important issue. As explained in the Introduction, SNLP can be tackled invoking the LSD estimator of \cite{ORTROMARA}, provided that they verify a {\em monotonizability} condition \cite[Assumption A1]{ORTROMARA}. This issue is discussed in Subsection \ref{subsec23} below. \\
   
\noindent {\bf D2} Notice that the parameter $\theta_5$ does not appear in the SNLP. It turns out that, as explained in Lemma \ref{lem3} below, the generation of an estimate is a trivial task with an elementary algebraic computation.\\
   
\noindent {\bf D3} A potential drawback of the result above is the need to reconstruct the signal $\dot u$. In some practical applications the PEMFC current acts as a control signal \cite{XINetal}, and it is often designed including an integral action. In this case,  $\dot u$ is known. Another practical solution is to approximate it using a dirty derivative filter ${p \over \tau p+1}$ that, with small values of the parameter $\tau>0$, gives a precise reconstruction of the signal derivative. Moreover, in Appendix A we show that this assumption can be {\em dispensed} when the PEMFC current is of the form $u=A \sin (\omega t)$, with $A$ and $\omega$ {\em unknown}, which is a scenario that appears in {\em off-line} estimation tasks.\\

\noindent {\bf D4}  We underscore the---often overlooked---fact that the filtering action $z={\lambda p \over p+\lambda}(v)$, which appears in the definition of the regressor \eqref{phim1} and in \eqref{ym1}, can be implemented {\em without differentiation} via the realization 
\begalis{
\dot x&=-\lambda(x+\lambda v)\\
z&=x+\lambda v
}
%
\subsection{Construction of a strictly monotonic mapping}
\lab{subsec23}
%
In this subsection we construct a new mapping verifying the monotonicty condition required by the LSD algorithm \cite{ORTROMARA}.

\begin{lemma}
\lab{lem2}\em
Consider the mapping $\calw(\theta_{1,4})$ given in  \eqref{calw}. Define the injective mapping
$$
\theta_{1,4}=\cald(\eta):=\begmat{\eta_1 \\ {\eta_3 \over \eta_2} \\  \eta_2 \\  {1 \over \eta_2}\Big( {\eta_3 \over \eta_2}-\eta_4 \Big)}, 
$$
whose inverse is given as
$$
\eta=\begmat{\theta_1 \\ \theta_3 \\ \theta_2\theta_3 \\ \theta_2-\theta_3 \theta_4}.
$$
The composed mapping $\calg:\rea^4 \to \rea^5$  defined as
$$
\calg(\eta) :=\calw(\cald(\eta)).
$$
satisfies the LMI
\begequ
\lab{demcon1}
T\nabla \calg(\eta) + [\nabla \calg(\eta)]^\top T^\top \geq \rho I_4 >0,
\end{equation}
for the matrix
$$
T:=\begmat{1 & 0 & 0 & 0 & 0\\ 0 & 1 & 0 & 0 & 0\\  0 & 0 & 0 & 1 & 0\\  0 & 0 & 0 & 0 & 1}.
$$
Hence, the mapping $T\calg(\eta)$ is monotonic, as required in \cite[Assumption A1]{ORTROMARA}
\end{lemma}

\begin{proof}
The mapping $\calg(\eta)$  is given by
$$
\calg(\eta)=\begmat{\eta_1 \\ \eta_2 \\ \eta_1 \eta_2 \\ \eta_3 \\ \eta_4},
$$
from which it is clear that \eqref{demcon1} is satisfied. 
\end{proof}
\subsection{Algebraic parameter estimator of $\theta_{5}$}
\lab{subsec24}
%
As indicated above, the parameter $\theta_5$ {\em does not} appear in the SNLP, hence it has to be estimated separately. This is easy to be done using the estimated vector $\hat \theta_{1,4}$, and is described in the proposition below. 
\begin{lemma}\em
\lab{lem3}
Consider the polarization curve \eqref{yu}, \eqref{yu1} verifying Assumption {\bf A1}. Define the estimated parameter
\begali{
{\hat \theta}_5 &= e^{-\hat \theta_3 u} (y - \hat \theta_4-\hat \theta_1 \ln(u)-\hat \theta_2 u)
\lab{hatthe5}
} 
where the estimates $\hat \theta_{1,4}$ converge exponentially to their true values $\theta_{1,4}$. Then,
\begequ
\lab{expcon1}
\liminf |\tilde \theta_{5}(t)|=0\quad (exp),
\endequ
where $\tilde {(\cdot)}:=\hat {(\cdot)} - (\cdot)$  are the estimation errors.
\end{lemma}
\begin{proof}
The proof follows trivially writing \eqref{hatthe5} in terms of the parameter errors $\tilde \theta_{1,4}$ as
$$
{\hat \theta}_5 = e^{-(\tilde \theta_3 +\theta_3) u} [y - (\tilde \theta_4 + \theta_4) -(\tilde \theta_1 +\theta_1) \ln(u)-(\tilde \theta_2 +\theta_2)u],
$$
and noting that
\begalis{
& \liminf \left\{e^{-(\tilde \theta_3(t) +\theta_3) u} [y - (\tilde \theta_4(t) + \theta_4) -(\tilde \theta_1(t)+\theta_1)\ln(u)-(\tilde \theta_2(t) +\theta_2)u]\right\}\\
& =e^{-\theta_3 u} [y - \theta_4 -\theta_1 \ln(u)-\theta_2)u] = \theta_5.
}
\end{proof}
\subsection{On the linear dependence of the regressor elements in \eqref{phim1}}
\lab{subsec25}
%

It can be numerically observed that the the elements of $\Phi_{M1}$ in \eqref{nlpre1}  exhibit linear dependence. This fact is  the major drawback of \eqref{nlpre1} since it makes $\mathcal{W}(\theta)$ non-identifiable. To evidence such a behaviour, we consider the model equation \textbf{M1} with $\theta_1=-2.582$, $\theta_2=-0.1808$, $\theta_3=0.0046$, $\theta_4=39.3543$ and $\theta_5=-1.2610$, and perform a numerical simulation---the corresponding polarization curve is shown in Fig. \ref{sim_lindep1}. We start it with a first test (Test 1) in which we consider a current signal $$u=25+5\cos(0.2\pi t)\mathrm{A}.$$ Fig. \ref{sim_lindep2} plots this signal and the resulting output voltage $y$. In order to prove linear dependence, the determinant of the Wronskian matrix of $\phi_{M1}$ is computed during the simulation time. Namely, we compute the determinant of the $5\times 5$ matrix
\begin{equation}\label{det}
W(\phi_{M1}(t))=\begin{bmatrix}\phi^\top_{M1}(t)\\p\Big(\phi^\top_{M1}(t)\Big)\\\vdots\\p^4 \Big(\phi^\top_{M1}(t)\Big)\end{bmatrix}.
\end{equation}
Simulation results of the $\mathrm{det}\{W(\phi_{M1}(t))\}$ are depicted in Fig. \ref{sim_lindep4}. Due to the filters' initial conditions, it can be noticed that there is a transient in which the determinant is non-zero. However, as fast as this transient vanishes, the determinant  is zero. Using the Wronskian test of linear dependence---see, for example, \cite[p.10]{poole}---it follows, from the zero determinant, that the elements of  $\Phi_{M1}$ are linearly related.

Afterwards, a new test (Test 2) is carried out in which we set an input current $$u=25+\frac{20}{\pi} \left[ \sin(0.2\pi t)  + \frac{1}{3}\sin(0.6 \pi t) +  \frac{1}{5}\sin(\pi  t)  \right]\mathrm{A}.$$ That is, compared with Test 1, we have included two more sine signals in $u$---more precisely, to synthesized $u$, we add a constant to the three first terms of the Fourier series decomposition of a pulse train. The pair of signals $u$ and $y$ of Test 2 is shown in Fig. \ref{sim_lindep2}. The five components of $\Phi_{M1}$ have been plotted in Fig. \ref{sim_lindep3}. It is immediately noticeable that all the curves in the figure exhibit the same shape. Thus, it is possible to get one curve of the plot from other by an appropriate scaling factor---that is, the elements in $\phi_{M1}$ are linearly dependent as in Test 1.

\begin{figure*}[t!]
	\centering
    \begin{subfigure}[Polarization curve using \textbf{M1}.]
		{\includegraphics[width=0.49\linewidth]{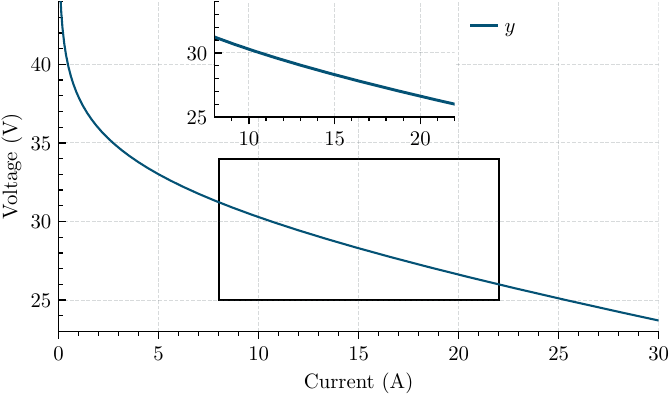}\label{sim_lindep1}}
	\end{subfigure}\hfill
	\begin{subfigure}[$u$ and $y$ obtained from Test 1 and Test 2.]
		{\includegraphics[width=0.49\linewidth]{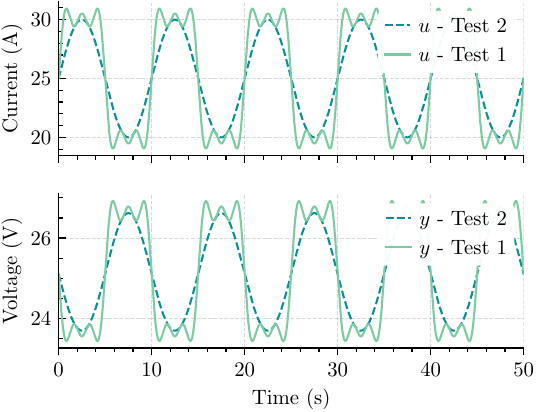}\label{sim_lindep2}}
	\end{subfigure}\\
 \begin{subfigure}[$\mathrm{det}\{W(\cdot)\}$.]
		{\includegraphics[width=0.49\linewidth]{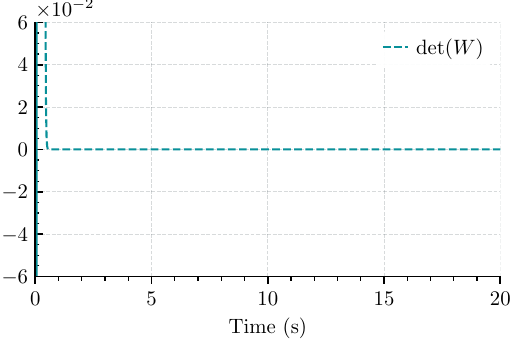}\label{sim_lindep4}}
	\end{subfigure}\hfill
    \begin{subfigure}[Elements of vector $\phi_{M1}\in\mathbb{R}^5$.]
		{\includegraphics[width=0.49\linewidth]{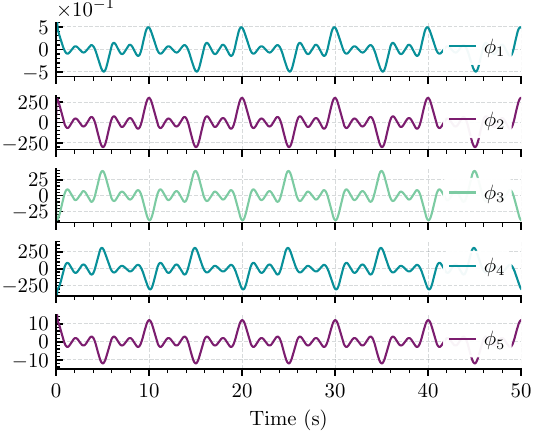}\label{sim_lindep3}}
	\end{subfigure}
	\caption{Simulation result of the linear dendence of $\phi_{M1}(t)$. }\label{sim_results}
\end{figure*}

%
\section{Separable Parameterizations for the Reduced Models M2-M3}
\lab{sec3}
%
In this section we derive separable parameterizations for the approximate models {\bf M2}-{\bf M3} of the polarization curve, which are given by the NSNLP \eqref{m2}, \eqref{m3} and \eqref{m4}, respectively. Interestingly, for the model {\bf M2} we will  derive a {\em linear} regression equation (LRE) that contains {\em both} uncertain parameters $a$ and $b$. Moreover, the construction does not involve any dynamic extension and is carried out with elementary {\em algebraic} operations. The derivation for the model {\bf M3} is also done with simple algebraic manipulations and also leads to an LRE, but requires a second order dynamic extension and it estimates only $b$. Consequently, in contrast with the case of {\bf M2}, it requires the application of a certainty-equivalent based calculation for the parameter $a$. Finally, the model {\bf M4} is clearly already  in LRE form, and the details are given here only for ease of reference.  
\subsection{Reduced Model M2}
\lab{subsec31}
%
\begin{lemma}\em
\lab{lem4}
Consider the model of the fuel cell voltage-current static characteristic given in \eqref{m2} with {\em known} $E_{oc}$ satisfying \eqref{eoc1}. The following LRE holds
\begin{equation}
\label{lrem2}
Y_{M2}(t)= \phi_{M2}^\top(t) \begmat{\ln(a) \\ b},\quad \forall t \geq 0,
\end{equation}
where we defined 
\begali{
\lab{ym2}
Y_{M2} &:=\ln(E_{oc}-v_{fc})\\
\lab{phim2}
\phi_{M2} &:=\begmat{1 \\ i_{fc}}.
}	
\end{lemma}
\begin{proof}
Consider the function  \eqref{m2} which may be written as
\begalis{
E_{oc}-v_{fc} &  = a e^{b i_{fc}} = e^c e^{b i_{fc}} = e^{c+b i_{fc}}
}
where we defined the constant 
\begequ
\lab{c}
c:=\ln(a).
\endequ 
Using the definition \eqref{ym2} we can write the equation above as
$$
Y_{M2}= c+b i_{fc},
$$ 
that, together with \eqref{phim2} and \eqref{c} yields \eqref{lrem2}, completing the proof. 
\end{proof}
\subsection{Reduced Model M3}
\lab{subsec32}
%
 \begin{lemma}\em
\lab{lem5}
Consider the model of the fuel cell voltage-current static characteristic given in \eqref{m3} with {\em known} $E_{oc}$ satisfying \eqref{eoc1}. Define the following dynamic extension
\begali{
\nonumber
\dot x_1 & =-\lambda [x_1 - \lambda \ln(E_{oc} - v_{fc})],\; x_1(0)= \lambda\log(E_{oc} - v_{fc}(0))\\
\lab{dotx}
\dot x_2 & =-\lambda [x_2 - \lambda \ln(i_{fc})],\; x_2(0)={\lambda}\ln({i_{fc}}(0)),
}	
with tuning gain the scalar $\lambda>0$. The following LRE holds:
\begin{equation}
\label{lrem3}
Y_{M3}(t)= \phi_{M3}(t) b,\quad \forall t \geq 0,
\end{equation}
where we defined the signals
\begali{
\lab{ym3}
Y_{M3} & :=x_1- \lambda \ln(E_{oc} - v_{fc})\\
\lab{phim3}
\phi_{M3} & := x_2 - \lambda \ln(i_{fc}).
}	
\end{lemma}
\begin{proof}
Define the signal
$$
z:=E_{oc}-v_{fc},
$$
whose differentiation, using \eqref{m3}, yields
\begalis{
\dot z & = a b (i_{fc})^{b-1} \dot{\arc {i_{fc}}} = b {\dot{\arc {i_{fc}}}\over i_{fc}}\left[a(i_{fc})^b\right]  = b {\dot{\arc {i_{fc}}}\over i_{fc}}z,
}
where we used again  \eqref{m3} to get the third identity. From the equation above we get 
$$
{\dot z \over z} = b {\dot{\arc {i_{fc}}}\over i_{fc}}.
$$
Defining now the signals
\begalis{
w & := \ln(z), \;m  := \ln(i_{fc}),
}
which satisfy
$$
\dot w = b \dot m.
$$
Applying to this equation the filter \eqref{calf} yields
\begequ
\lab{filsig}
{\lambda p\over p+\lambda}(w)=b{\lambda p\over p+\lambda}(m).
\endequ
The proof is completed noting that \eqref{dotx} is a state-space realization of \eqref{ym3}, \eqref{phim3}, \eqref{filsig}.
\end{proof}
\subsection{Reduced Model M4}
\lab{subsec33}
%
It is clear from \eqref{m4} that in this case it is trivial to derive a LRE as follows:
\begin{lemma}\em
\lab{lem6}
Consider the model of the fuel cell voltage-current static characteristic given in \eqref{m4}. The following LRE holds
\begin{equation}
\label{lrem4}
Y_{M4}(t)= \phi_{M4}^\top(t) \begmat{\theta_6 \\ \theta_1\\ \theta_2},\quad \forall t \geq 0,
\end{equation}
where we defined 
\begalis{
Y_{M4} &:=v_{fc}\\
\phi_{M4} &:=\begmat{1 \\ \ln(i_{fc})\\i_{fc}}.
}	
\end{lemma}
\subsection{Discussion}
\lab{subsec34}

\noindent {\bf D5}  The simplicity of the proof of Lemma \ref{lem4} for the model {\bf M2} can hardly be overestimated. A potential drawback of this result is the fact that that the regressor $\phi_{M2}$ of the LRE \eqref{lrem2} contains a {\em constant} term that may jeopardize the excitation requirement. However,  it is possible to use the LSD estimator of \cite{ORTROMARA} and impose a simple IE condition preserving the exponential converge property. Since the simulation results with the gradient were satisfactory we will stick to the gradient estimator avoiding the additional complication of the LSD estimator.\\

\noindent {\bf D6} The proof of Lemma \ref{lem5} for the model {\bf M3} is also extremely simple. The only difficulty been that we can estimate with the LRE \eqref{lrem3} only $b$. However, we can use the relation \eqref{m3} to estimate---in a certainty-equivalent way---the parameter $a$ as follows
$$
{\hat a}=(E_{oc}-v_{fc})\big(i_{fc}\big)^{-\hat b}.
$$

\noindent {\bf D7} In \eqref{dotx} we have selected the initial conditions of the state equations in such a way as to avoid an additive exponentially decaying term $\et$ in the regression \eqref{lrem3}. Although, as discussed in Footnote 2, this term does not affect the asymptotic behavior of the estimator, it affects the transient.  
%
\section{The LSD Estimator and Gradient Estimators}
\lab{sec4}
%
For the sake of completeness in this section we present the LSD estimator of \cite{ORTROMARA} and the classical gradient scheme \cite{SASBODbook,TAObook}. They are given in a generic style to be applied to NLPRE of the form 
\begequ
\lab{lsd}
Y=\psi^\top \calw(\eta),
\endequ 
with $Y(t)\in \rea$, $\psi(t) \in \rea^p$, $\eta \in \rea^q$ and $\calw:\rea^q \to \rea^p$, or the LRE 
\begequ
\lab{gra}
Y=\psi^\top \eta,
\endequ 
with $p=q$.
\subsection{The LSD parameter estimator}
\lab{subsec41}
%
We present in this subsection the LSD  estimator for the NLPLRE \eqref{lsd}. The proof of the proposition below is given in \cite[Proposition 1]{ORTROMARA}, therefore it is omitted here.\\

For the LSD estimator we need to impose the following IE assumption \cite{TAObook} on the regressor  $\phi$.\\

\noindent {\bf A2} The regressor vector $\phi$ given in \eqref{lsd} is IE. That is, {there exist} constants $C_c>0$ and $t_c>0$ such that
$$
\int_0^{t_c} \phi(s) \phi^\top(s)  ds \ge C_c I_p.
$$

\begin{proposition}\em
\lab{pro2}
Consider the NLPRE \eqref{lsd}, with the regressor vector $\phi(t)$ satisfying Assumption {\bf A2}. Define the LSD estimator
\begalis{
\dot{\hat \calw} & =\gamma_0 F \phi (Y-\phi^\top {\hat \calw} ),\; \hat\calw(0)=\calw_{0} \in \rea^{p}\\
\dot {F}& =  -\gamma_0 F \phi  \phi^\top  F,\; F(0)={1 \over f_0} I_p \\
\dot{\hat \eta} & =  \Gamma  \Delta T [\caly -\Delta \calw(\hat \eta) ],\; \hat \eta(0)=\eta_0 \in \rea^q,
}	
where the mapping $\calw(\eta)$ satisfies the LMI
$$
T\nabla \calw(\eta) + [\nabla \calw(\eta)]^\top T^\top \geq \rho I_q >0,
$$
for some $T \in \rea^{q \times p}$ and we defined the signals
\begalis{
\Delta & :=\det\{I_p-f_0F\} \in \rea\\
\caly & := \adj\{I_p- f_0F\} (\hat\eta -  f_0F \eta_{0}) \in \rea^{p},
}
where $ \adj\{\cdot\}$ denotes the adjugate matrix, and the tuning gains are the scalars  $f_0>0$, $\gamma_0>0$ and the positive definite, diagonal matrix $\Gamma$. For all initial conditions $\calw_{0} \in \rea^p$ and $\eta_{0} \in \rea^q$ we have that
\begequ
\lab{expcon}
\liminf |\tilde \eta(t)|=0,\; (exp),
\endequ
 with all signals {\em bounded}.,
\end{proposition}
\subsection{Gradient parameter estimator}
\lab{subsec42}
%
We present in this section the classical gradient estimator for the LRE \eqref{gra}. The proof of the proposition below is very standard and may be found in any book of adaptive control, {\em e.g.} \cite[Theorem 2.5.3]{SASBODbook}, therefore it is omitted here.

For the gradient estimator we need to impose the following assumption \cite{TAObook} on the regressor  $\psi$.\\

\noindent {\bf A3} The regressor $\psi$  is {\em persistently exciting}, that is, there exists $T>0$ and $\rho>0$ such that
$$
\int_t^{t+T} \psi(s) \psi^\top(s) ds = \rho I_q, \; \forall t \geq 0,
$$
 
\begin{proposition}\em
\lab{pro3}
Consider the LRE \eqref{gra} with $\psi$ verifying Assumption {\bf A3}, and the gradient estimator
\begalis{
\dot {\hat \eta} & =-\Gamma \psi (\psi^\top {\hat \eta}  -Y),
}
with $\Gamma \in \rea^{q \times q}$ positive definite. Under these conditions, the parameter estimation errors verify \eqref{expcon}.
\end{proposition}
%
\section{Simulation Results}
\lab{sec5}
%

Simulations are conducted to test the estimation algorithms based on the previously introduced models \textbf{M2}-\textbf{M4}. The simulated system is the Matlab block corresponding to a $1.26$kW PEMFC stack of the Simscape/Matlab electrical Library. It consists of 42 cells with an open-circuit voltage of $E_{oc}=42$V. Also, current and voltage are related in a static manner in this PEMFC block. Therefore, it does not exhibit the hysteresis behavior usual in a physical setting. Fig. \ref{sim_curvepemfc} shows the polarization curve of the block. For the presented simulation tests, the stack is excited with a pulse train varying from $10$A to $20$A at a frequency of $2$Hz. This yields output voltage changes from $30.006$V to $28.051$V---see Fig. \ref{sim_uy}. 

First, we assessed the performance of the algorithm derived from the LSD estimator for $\textbf{M2}$---see Prop. \ref{pro2} and Lemma \ref{lem3}, respectively. It was realized with the estimator constants: $\gamma_0$=24, $f_0=1\times10^{-5}$, $\gamma_1=3\times10^4$, and $\gamma_2=3\times10^4$. From Fig. \ref{sim_M2}, it can be observed the evolution of $\hat a$ and $\hat b$ obtained during the simulation. Note that $\hat a$ converges to $10.3136$ whereas $\hat b$ tends to $0.0151$.  Fig. \ref{sim_polcurves} depicts the estimated polarization curve plotted from those estimated values. It can be compared with the polarization curve of the PEMFC system, also included in the figure.

Second, simulation results of the LSD estimator based on $\textbf{M3}$---see Lemma \ref{lem4}---were performed. For this test, we set $\lambda=80$, $\lambda_a=7$, $\lambda_b=7$, $\gamma_0$=24, $f_0=1\times10^{-5}$, and $\gamma_1=600$. To visualize the convergence,  $\hat a$ and $\hat b$ are plotted as a function of time in Fig. \ref{sim_M3}. At the stop time of simulation $\hat a=7.2641$ and  $\hat b=0.2178$.  The estimated polarization curve obtained from those values is displayed in Fig. \ref{sim_polcurves}.    

Third, a simulation test was carried out to evaluate the performance of the  LSD estimator with $\textbf{M4}$---see Lemma \ref{lem5}. For this simulation we take $\gamma_0=24$, $f_0=1\times 10^{-6}$, $\gamma_1=30$, $\gamma_2=30$, and $\gamma_3=30$. Fig. \ref{sim_M4}  shows how the estimations of $\hat\theta_1$, $\hat\theta_2$ and $\hat\theta_6$ evolve during the simulation. These converge to the values: $\hat\theta_1={-1.9271}$, $\hat\theta_2={-0.0619}$, and $\hat\theta_6={35.0619}.$ As previously, the estimated curve resulting from this test is shown in Fig. \ref{sim_polcurves}. 

As a final remark, it can be observed from \ref{sim_polcurves} that \textbf{M2} represents appropriately the model for currents within the interval $u\in[10\mathrm{A}, 20\mathrm{A}]$, for $u<10$A and $u>20$ the difference is noticeable and increasing.
On the other hand, \textbf{M3} represents appropriately the model for the interval $u\in[1\mathrm{A}, 30\mathrm{A}]$. For currents out of that range there is a barely noticeable difference. To conclude, Model \textbf{M4} represents remarkably well the model in the entire operation range and exhibits the best performance.

\begin{figure*}[h!]
	\centering
	\begin{subfigure}[Polarization curve of the Matlab's PEMFC block.]
		{\includegraphics[width=0.49\linewidth]{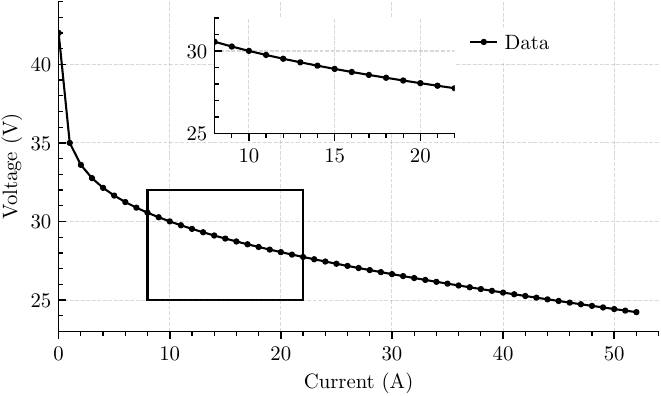}\label{sim_curvepemfc}}
	\end{subfigure}\hfill
	\begin{subfigure}[ Input current $u$ and output voltage $y$.]
		{\includegraphics[width=0.49\linewidth]{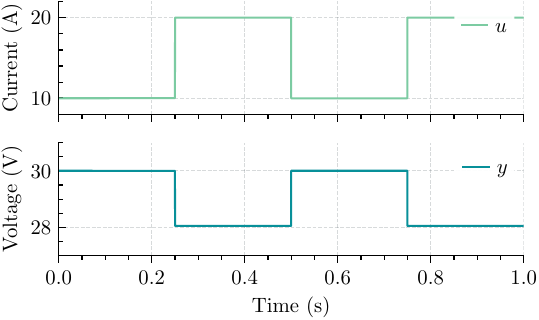}\label{sim_uy}}
	\end{subfigure}\hfill
	\begin{subfigure}[ Online estimation of $\hat a$ and $\hat b$ in \textbf{M2}.]
		{\includegraphics[width=0.49\linewidth]{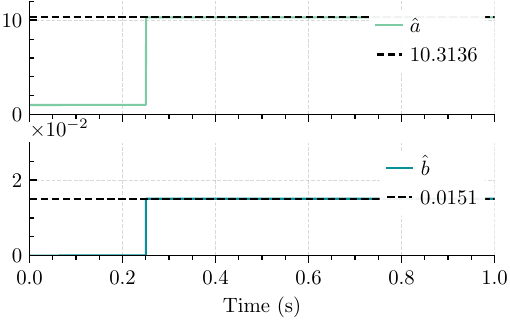}\label{sim_M2}}
	\end{subfigure}\hfill
	\begin{subfigure}[ Online estimation of $\hat a$ and $\hat b$ in \textbf{M3}.] 
		{\includegraphics[width=0.49\linewidth]{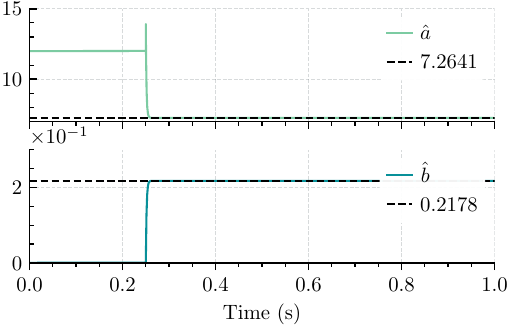}\label{sim_M3}}
	\end{subfigure}\hfill
	\begin{subfigure}[Online estimation of $\hat \theta_1$, $\hat \theta_2$ and $\hat \theta_6$ in \textbf{M4}.] 
		{\includegraphics[width=0.49\linewidth]{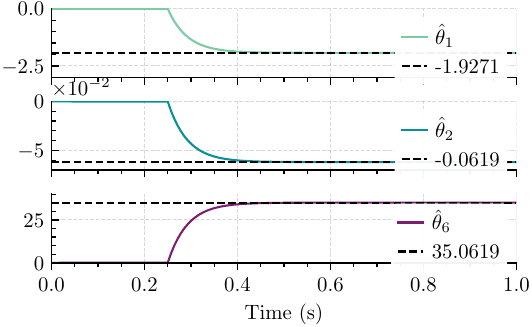}\label{sim_M4}}
	\end{subfigure}\hfill
	\begin{subfigure}[Polarization curve of the simulated PEMFC and the estimated polarization curves \textbf{M2}-\textbf{M4}. ] {\includegraphics[width=0.49\linewidth]{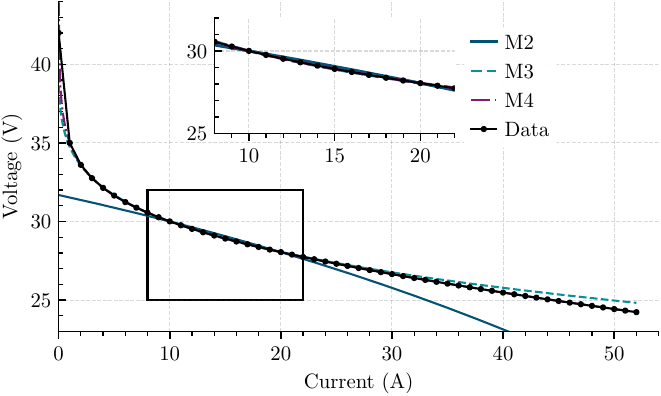}\label{sim_polcurves} }
	\end{subfigure}
	\caption{Simulation results of the online estimation. }\label{sim_results}
\end{figure*}

%
\section{Experimental Results}
\lab{sec6}
%
Experimental tests are carried out for experimental validation of the previously introduced LSD estimator for models \textbf{M2}-\textbf{M4} in a fully automated Nexa PEMFC stack. This PEMFC stack consists of 65 cells, a rated output power of 1.2kW, and a nominal output voltage of 26V. Moreover, the following equipment is used to perform such tests: an opamp-based data acquisition system for PEMFC current and voltage measurements, an electronic variable load, and the DSPACE DS1104 for real-time numerical computations. This test bench is depicted in Fig. \ref{testbench}.

\begin{figure}[b]
	\begin{center}
		\includegraphics[width=0.4\textwidth]{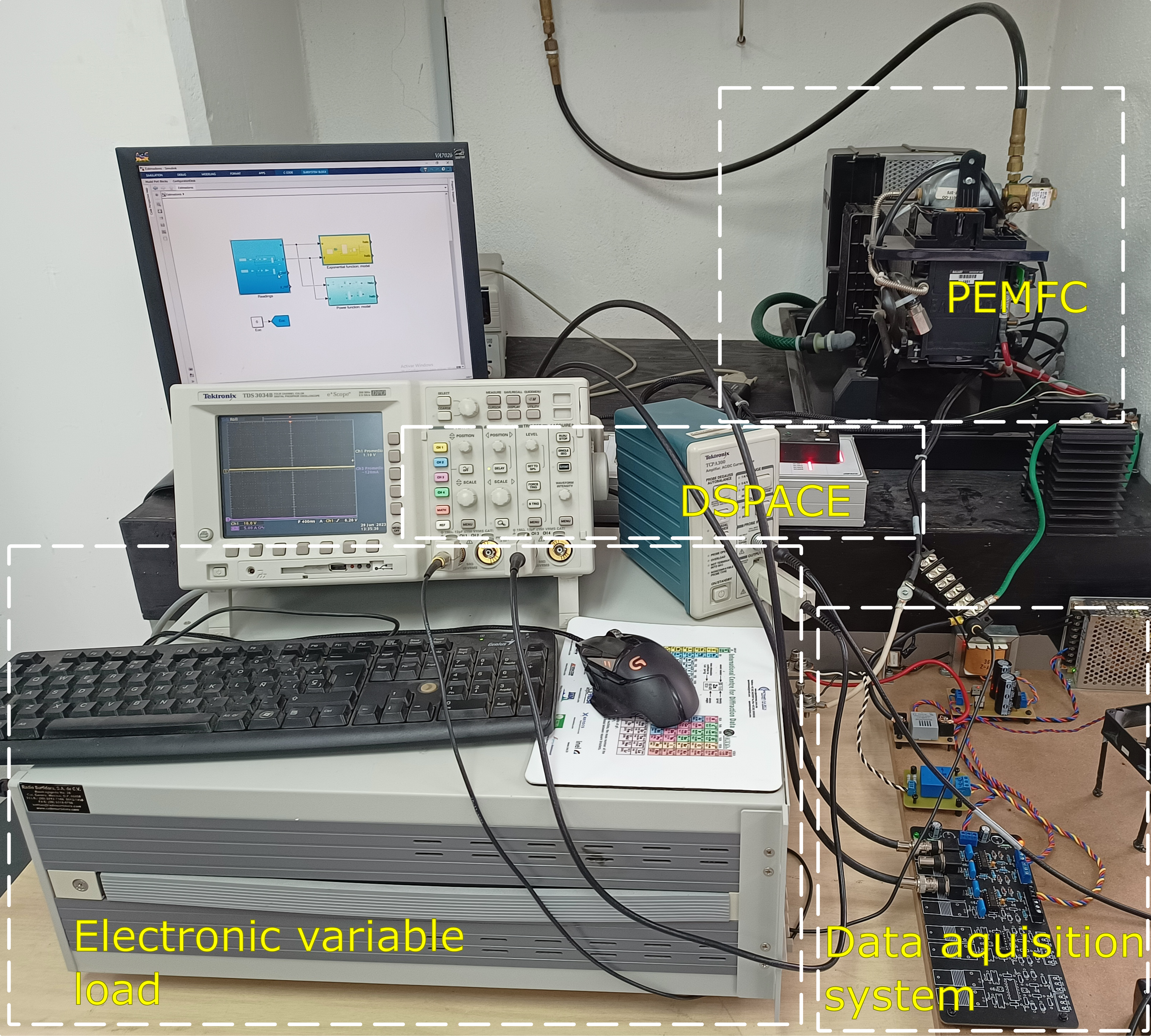}
		\caption{Test bench using 1.2 $kW$ Nexa PEMFC power module.}\label{testbench}
	\end{center}
\end{figure}

\noindent {The polarization curve of the fuel cell was experimentally obtained by driving the cell's current to several operation points and  then registering the  corresponding steady-state voltage for each point. The hysteresis phenomenon causes two different paths in the plot depending on whether the subsequent current operation points have a rise or fall in their value. One path in the hysteresis loop of Fig. \ref{exp_pemfc} is produced when the current is varied upward from $0$A to $30$A. The other path is produced when the current is varied downward  from $30$A to $0$A. These two paths have been averaged to obtain a new curve to which we refer in the sequel as the ``polarization curve''---see Fig. \ref{exp_pemfc}. Moreover, the open-circuit voltage of the stack, measured before testing, was $E_{oc}=39.8$V.

\subsection*{Approximation of the polarization curve using curve fitting}

\noindent As a starting point, we approximate the behavior of the polarization curve in Fig. \ref{exp_pemfc} by means of a curve fitting procedure applied to each one of the previously introduced models \textbf{M2}-\textbf{M4}. Fig. \ref{exp_fit} shows these approximations from computations using Matlab's Curve Fitting Toolbox \cite{CurveFitting}.  The obtained fitting parameters for each model are the following.  
For \textbf{M2}, $a=4.52$ and $b =0.0463$.
For \textbf{M3}, $ a =2.117$ and $b =0.5921$. 
And for \textbf{M4}, $\theta_1=-0.7984$, $\theta_2=-0.3709$, and $\theta_6=37.31$.

\begin{figure*}[t]
	\centering
		\includegraphics[width=0.5\linewidth]{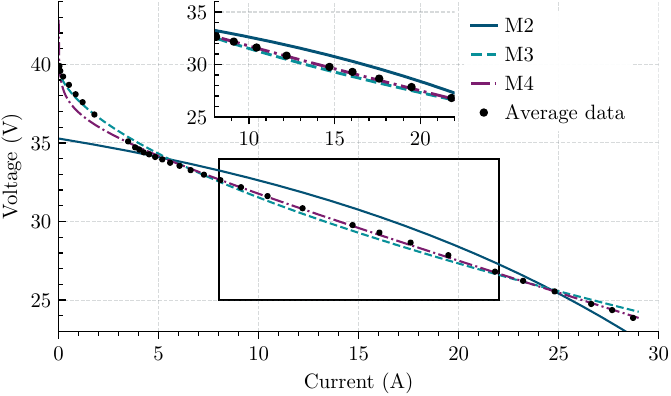}
	\caption{Polarization curve of the PEMFC system and the approximated polarization curves using curve fitting applied to \textbf{M2}-\textbf{M4}. }
	\label{exp_fit}
\end{figure*}

\subsection*{Experimental results of the online estimation}

We now present the results of the online estimation using the three proposed algorithms. The input current for the three experimental tests is a pulse train varying from $10\mathrm{A}$ to $20\mathrm{A}$ with a period of 3.3 s---see Fig. \ref{exp_uy}. To reduce the noise produced when sensing, high frequencies were filtered out by a low-pass filter from the input current and output voltage signals. Fig. \ref{exp_uy} shows each sensed signal and its filtered version used to run the estimation algorithms. All the tests below were carried out with the LSD estimator.

The first experimental test consisted in assessing the performance of the algorithm derived from $\textbf{M2}$. {This result was obtained under the following estimator gains: ${f_0=0.1}$,  $\gamma_0=5$, and $\Gamma=\mathrm{diag}(30,~30)$.}
Fig. \ref{exp_M2} shows $\hat a$ and $\hat b$ obtained during the simulation time.  Note from the figure that $\hat a$ converges to $5.3842$ and $\hat b$ to $0.0438$.   The estimated polarization curve plotted from those values is presented in Fig. \ref{exp_polcurves} together with the actual polarization curve of the PEMFC system.

Next, results of the estimation based on $\textbf{M3}$ were performed. 
Fig. \ref{exp_M3} displays signals $\hat a$ and $\hat b$ obtained as a result of this test.  We set ${f_0=0.1}$,  $\gamma_0=2.5$, and $\Gamma=6$.  Notice from the figure that the estimated values are $\hat a=2.1532$ and  $\hat b=0.5936$.  The  corresponding estimated polarization curve is displayed in Fig. \ref{exp_polcurves}.    

The final experimental results showing the estimates $\hat\theta_1$, $\hat\theta_2$ and $\hat\theta_6$ of  the LSD estimator for \textbf{M4} are shown in Fig. \ref{exp_M4}. {We fixed ${f_0=6\times 10^{-3}}$,  $\gamma_0=1.115$, and $\Gamma=\mathrm{diag}(0.5,0.5,0.5)$.}  
As it can be seen, $\hat\theta_1$ tends to $-1.2193$ whereas $\hat\theta_2$ tends to $-0.4233$ and $\hat\theta_6$ to $39.6861$. The estimated polarization curve resulting from this test is shown in Fig. \ref{exp_polcurves}. 

It can be observed from \ref{exp_polcurves} that, compared with the other two models, \textbf{M3} fits better the actual polarization curve in the entire operation interval. On the other hand, the reliability of \textbf{M2} is limited to the interval $u\in [5\mathrm{A},22\mathrm{A}]$. For current values out of that interval  the mismatch of \textbf{M2} with respect to the actual polarization curve is evident. On the other hand, \textbf{M4} approximate better for current levels above $7$A and for values close to $0$A, this model is clearly more reliable than \textbf{M2}. 

\begin{figure*}[h!]
	\centering
	\begin{subfigure}[Current-voltage curves of the PEMFC system: the hysteresis loop and the polarization (average) curve.]
		{\includegraphics[width=0.49\linewidth]{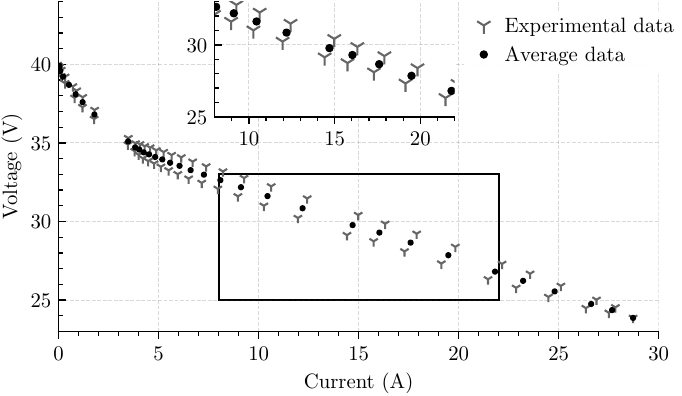}\label{exp_pemfc}}
	\end{subfigure}\hfill
	\begin{subfigure}[ Input current $u$ and its output voltage $y$.]
		{\includegraphics[width=0.49\linewidth]{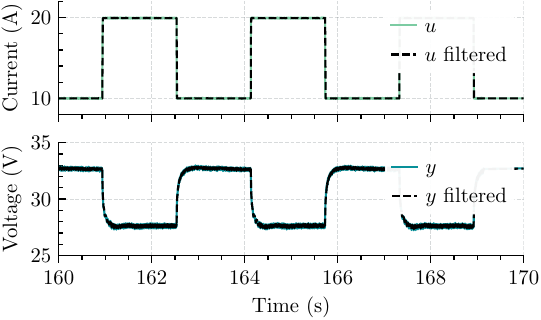}\label{exp_uy}}
	\end{subfigure}\hfill
	\begin{subfigure}[Online estimation of $\hat a$ and $\hat b$ in \textbf{M2}.]
		{\includegraphics[width=0.49\linewidth]{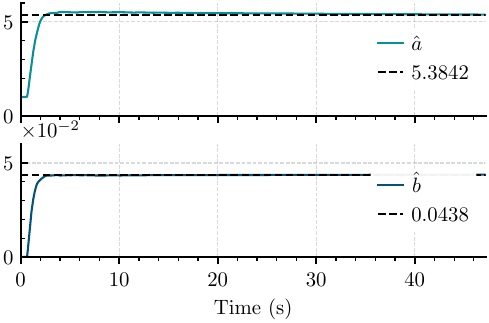}\label{exp_M2}}
	\end{subfigure}\hfill
	\begin{subfigure}[ Online estimation of $\hat a$ and $\hat b$ in \textbf{M3}.] 
		{\includegraphics[width=0.49\linewidth]{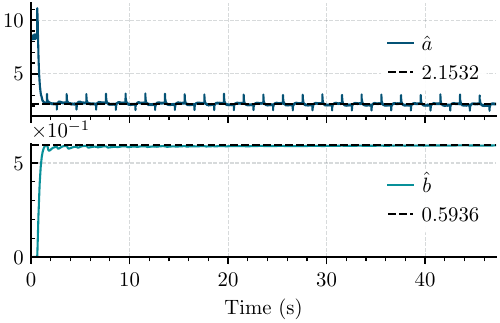}\label{exp_M3}}
	\end{subfigure}\hfill
	\begin{subfigure}[ Online estimation of $\hat \theta_1$, $\hat \theta_2$ and $\hat \theta_6$ in \textbf{M4}] 
		{\includegraphics[width=0.49\linewidth]{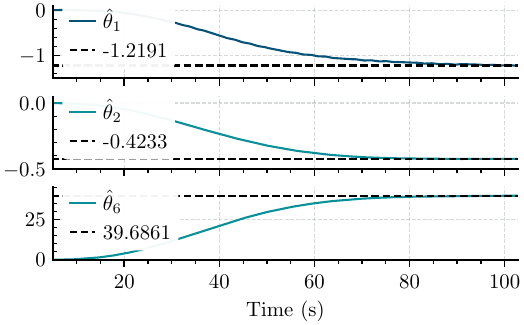}\label{exp_M4}}
	\end{subfigure}\hfill
	\begin{subfigure}[Polarization curve of the PEMFC system and the estimated polarization curves \textbf{M2}-\textbf{M4}. ] 
		{\includegraphics[width=0.49\linewidth]{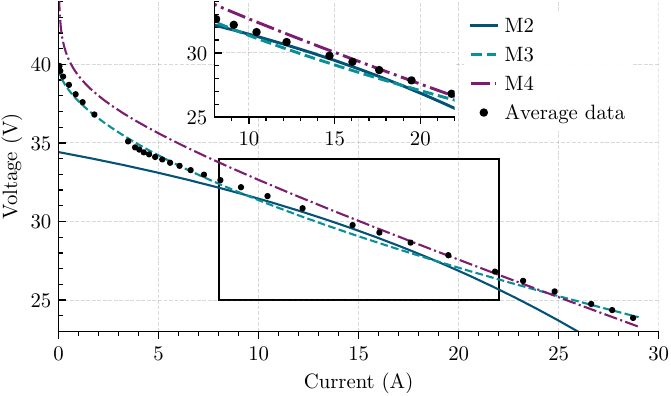}\label{exp_polcurves} }
	\end{subfigure}
	\caption{Experimental results of the online estimation. }\label{exp_results}
\end{figure*}

\section{Concluding Remarks and Future Research}
\lab{sec7}
%
We have presented the first solution to the problem of designing a GEC, on-line algorithm to estimate the parameters of the full model {\bf M1} of the polarization curve of a PEMFC that is described by a NSNLP. The design proceeds in three steps: (i) the derivation of a SNLP for the parameters $\theta_{1,4}$; (ii) the application of the LSD estimator to reconstruct these parameters from this SNLP; (iii) the application of simple algebraic operation to generate the estimate of $\theta_5$. It should be underscored that the use of the LSD estimator is necessary to deal with the non-linear nature of the regression equation (without overparameterization). Unfortunately, it is numerically shown that the elements of the SNLP regression are {\em linearly dependent}. Hence, the vector of unknown parameters is not-identifiable. It should be underscored that this negative result does not imply that the original NSNLP \eqref{yu1} is not identifiable.  

Besides the estimator for the full model,  we have presented extremely simple procedures to estimate the parameters of the approximate models {\bf M2}-{\bf M4}---which are widely used for control of PEMFC. In these cases, it is possible to derive a LRE from which the unknown parameters can be estimated with standard algorithms. 

Simulation and experimental results to evaluate the performance on the estimator of the full model in the face of parameter variations and noise have been reported. Also, a study has been carried out to assess the ability of the reduced models to approximate the behavior of the exact polarization curve.

 Current research is under way to use these estimator results in the context of adaptive control. Also, in the spirit of \cite{XINetal}, it is of interest to study the variations with temperature of the reversible potential, which is assumed constant in this paper. Some preliminary results on this topic have been reported in \cite{ORTetalwc23}.  
    

%
\appendices
\section{Case of $i_{fc}(t)=A \sin(\omega t)$ with $A$ and $\omega$ Unknown}
\lab{appa}
%
In this appendix we consider the case when the current takes the form
$$
u(t)=A \sin(\omega t),
$$
with $A>0$ and $\omega>0$ unknown, which is of interest for off-line estimation. We will show below that, in this case, it is possible to remove the assumption of known $\dot u$. Instrumental to obtain this result is the Swapping Lemma  \cite[Lemma 3.6.5]{SASBODbook}.

Towards this end, we recall equation \eqref{nlpre1}, that we rewrite here for ease of reference:
\begalis{
Y & =\theta_3 {\lambda \over p+\lambda}(\dot u y)- \theta_1\theta_3{\lambda \over p+\lambda}(\dot u \ln(u))+ \theta_1 {\lambda p \over p+\lambda}(\ln(u))+ \theta_2{\lambda p \over p+\lambda}(u)  - \theta_3\theta_4 {\lambda p \over  p+\lambda}(u) -{\theta_2\theta_3 \over 2} {\lambda p\over p+\lambda}(u^2).
}
We concentrate our attention on the terms depending on $\dot u$. First, notice that 
$$
{\lambda \over p+\lambda}(\dot u \ln(u))={\lambda p\over p+\lambda}(u \ln(u))-{\lambda p\over p+\lambda}(u).
$$
Second, note that we can invoke the Swapping Lemma to write 
\begalis{
{\lambda \over p+\lambda}(\dot u y)&=\dot u{\lambda \over p+\lambda}(y)-{\lambda \over p+\lambda}\left(\ddot u{1 \over p+\lambda}(y)\right)\\
&=\dot u{\lambda \over p+\lambda}(y)+{\lambda \over p+\lambda}\left(\omega^2 u{1 \over p+\lambda}(y)\right),
}
where we have used the fact that $\ddot u = - \omega^2 u$, in the second identity. Replacing these two identities in the NLPRE above we get
\begalis{
Y= & \theta_3 \dot u{\lambda \over p+\lambda}(y)+\theta_3\omega^2{\lambda \over p+\lambda}\left(u{1 \over p+\lambda}(y)\right)- \theta_1\theta_3\left[{\lambda p\over p+\lambda}(u \ln(u))-{\lambda p\over p+\lambda}(u)\right]\\
& + \theta_1 {\lambda p \over p+\lambda}(\ln(u))+ \theta_2{\lambda p \over p+\lambda}(u) - \theta_3\theta_4 {\lambda p \over  p+\lambda}(u)-{\theta_2\theta_3, \over 2} {\lambda p\over p+\lambda}(u^2)\\
 = &\theta_3 \left[\dot u{\lambda \over p+\lambda}(y)\right]+\begmat{ \calg^\top(\theta_{1,4}) & \theta_6}\xi,
}
with\footnote{Notice the, unlike the map $\calw(\theta_{1,4})$, the new map $\calg(\theta_{1,4})$ does not contain the term $\theta_3$ alone.}
\begequ
\lab{xi}
\calg(\theta_{1,4}):=\begmat{\theta_1 \\ \theta_2 \\ \theta_3 \theta_4 \\ \theta_1\theta_3 \\ \theta_2\theta_3} \in \rea^5,\;\xi:=\begmat{ {\lambda p \over p+\lambda}(\ln(u)) \\ {\lambda p \over p+\lambda}(u) \\ - {\lambda p \over  p+\lambda}(u)\\ -{\lambda p\over p+\lambda}(u \ln(u))+{\lambda p\over p+\lambda}(u)  \\ -{1 \over 2} {\lambda p\over p+\lambda}(u^2)\\{\lambda \over p+\lambda}\left(u{1 \over p+\lambda}(y)\right)
}\in \rea^6
\endequ
where we defined $\theta_6:=\theta_3 \omega^2$. Let us now look at the term $\dot u{\lambda \over p+\lambda}(y)$, to which we apply the filter $\calf(p)$ and invoke the Swapping Lemma to get
\begalis{
{\lambda \over p+\lambda}\left(\dot u{\lambda \over p+\lambda}(y)\right)&=\dot u {\lambda^2 \over (p+\lambda)^2}(y) - {\lambda \over p+\lambda}\left(\ddot u {\lambda \over (p+\lambda)^2}(y) \right)\\ 
& =\dot u {\lambda^2 \over (p+\lambda)^2}(y) + \omega^2{\lambda \over p+\lambda}\left( u {\lambda \over (p+\lambda)^2}(y) \right).
}
We now proceed to apply the filter $\calf(p)$ to $Y$ above yielding
$$
\calf(p)(Y) =\theta_3 \left[\dot u {\lambda^2 \over (p+\lambda)^2}(y)\right]+\theta_6 \left[ {\lambda \over p+\lambda}\left( u {\lambda \over (p+\lambda)^2}(y) \right) \right] +\begmat{ \calg^\top(\theta_{1,4}) & \theta_6}\calf(p)(\xi).
$$
To eliminate the term that depends on $\dot u$ we define the signal
\begalis{
Z &:=Y {\lambda^2 \over (p+\lambda)^2}(y)-\calf(p)(Y){\lambda \over p+\lambda}(y) \\
& = \begmat{ \calg^\top(\theta_{1,4}) & \theta_6}{\lambda^2 \over (p+\lambda)^2}(y)\xi- \begmat{ \calg^\top(\theta_{1,4}) & \theta_6}{\lambda \over p+\lambda}(y)\calf(p)(\xi) -\theta_6 \left[ {\lambda \over p+\lambda}\left( u {\lambda \over (p+\lambda)^2}(y) \right){\lambda \over p+\lambda}(y) \right]\\
& = \begmat{ \calg^\top(\theta_{1,4}) & \theta_6 }\chi,
}
where we defined the vector
$$
\chi:={\lambda^2 \over (p+\lambda)^2}(y)\xi-{\lambda \over p+\lambda}(y)\calf(p)(\xi) -{\bf e}_6{\lambda \over p+\lambda}\left( u {\lambda \over (p+\lambda)^2}(y) \right){\lambda \over p+\lambda}(y)
$$
with ${\bf e}_6:=\col(0,0,0,0,0,1)\in \rea^6$.

It is easy to show that the map $\begmat{ \calg^\top(\theta_{1,4}) & \theta_6}:\rea^5 \to \rea^6$ is monotonizable. Therefore, equipped with the SNLP $Z(t)=\begmat{ \calg^\top(\theta_{1,4}) & \theta_6}\chi(t)$ it is possible, using the LSD estimator, to estimate the parameters $\theta_{1,4}$ and $\theta_6$. 
%
\section{List of Acronyms}
\lab{appb}
%
{
\begin{table}[h]
	\centering
	\label{tab:2}
	\renewcommand\arraystretch{1.6}
	\begin{tabular}{l|r}
		\hline\hline
		DREM  &  Dynamic regressor extension and mixing \\		
        GEC & Global exponential convergence\\
		IE & Interval excitation \\	
		LRE & Linear regression equation \\		
		LSD & Least-squares plus DREM\\
		LTI &  Linear time-invariant  \\
		 SNLP  & Separable onlinear parameterization \\
	    NSNLP  & Nonseparable nonlinear parameterization \\
		PEMFC &  Proton exchange membrane fuel cell \\
		\hline\hline
	\end{tabular}
\end{table}
}
\end{document}